\documentclass[letterpaper, 10 pt, conference]{ieeeconf}  

\IEEEoverridecommandlockouts                              

\usepackage{amsmath} 
\usepackage{amssymb}  
\usepackage{graphicx}
\usepackage{subcaption}
\usepackage{color}
\usepackage{arydshln}

\newtheorem{prop}{Proposition}
\newtheorem{lem}{Lemma}
\newtheorem{remark}{Remark}
\DeclareMathOperator{\tr}{tr}
\DeclareMathOperator{\e}{\mathrm{e}}
\title{\LARGE \bf
Design of the Impulsive Goodwin's Oscillator: A Case Study$^\dagger$
}

\author{Alexander Medvedev$^{1}$, Anton V. Proskurnikov$^{2}$, and Zhanybai T. Zhusubaliyev$^{3,4}$
\thanks{$\dagger$ This manuscript extends our conference paper~\cite{MPZ22} and includes the proofs that have been omitted there.}
\thanks{* AM was partially supported the Swedish Research Council under grant 2019-04451. ZhZh was
partially supported  by the grant 14-22 of the Osh State University}
\thanks{$^{1}$Department of Information Technology, 
        Uppsala University, SE-752 37 Uppsala, Sweden
        [{\tt\small alexander.medvedev@it.uu.se}]}%
\thanks{$^{2}$Department of Electronics and Telecommunications, Politecnico di Torino, Turin, Italy, 10129 [{\tt\small anton.p.1982@ieee.org}]}%
\thanks{$^{3}$Department of Computer Science, International Scientific Laboratory for
Dynamics of Non-Smooth Systems, Southwest State University, Kursk, Russia
        [{\tt\small zhanybai@hotmail.com}]}%
\thanks{$^{4}$Faculty of Mathematics and Information Technology, Osh State University, Lenin st. 331, 723500, Osh, Kyrgyzstan        }%
}

\begin{document}

\maketitle
\thispagestyle{empty}
\pagestyle{empty}

\begin{abstract}
The impulsive Goodwin's oscillator (IGO) is a hybrid model
composed of a third-order continuous linear part and a pulse-modulated feedback. 
This paper introduces a design problem of the IGO to admit a desired periodic solution.
The dynamics of the continuous states  represent the  plant to be controlled, whereas  the parameters of the  impulsive feedback constitute design degrees of freedom.  The design objective is to select the free  parameters so that the  IGO exhibits a stable 1-cycle with desired characteristics. The impulse-to-impulse map of the oscillator is demonstrated to always possess a positive fixed point that corresponds to the desired periodic solution; the closed-form expressions to evaluate this fixed point are provided. Necessary and sufficient conditions for orbital stability of the 1-cycle are presented in terms of the oscillator parameters and exhibit similarity to the problem of static output control. An IGO design procedure is proposed and validated by simulation. The nonlinear dynamics of the designed IGO are reviewed by means of bifurcation analysis. Applications of the design procedure to dosing problems in chemical industry and biomedicine are envisioned. 
\end{abstract}

\section{INTRODUCTION}

In  control of engineered systems,  the  objective is normally to keep the controlled variable in a vicinity of a predefined setpoint or to make it follow a certain trajectory. In contast,  the purpose of physiological control is, arguably, to maintain the involved biological quantities within a certain domain, and to achieve this with minimal energy. 
Impulsive feedback control is one of the most widespread strategies applied by nature in physiological, especially in neuroendocrine, systems. In particular, the hypothalamic-pituitary adrenal and gonadal axes employ pulse-modulated control and encode information to target cells by manipulating both the amplitude and frequency of the hormone concentration pulses \cite{WTT10}.

The problem of exerting  a periodic  control action that maintains
a certain predefined level of effect in a dynamical plant often arises in process control and medicine. For instance, adding doses of chemicals to a reactor is typically done by means of logical (discrete) open-loop control~\cite{AH22}. Similarly, pharmaceuticals, in a tablet or an injection form, are predominantly administered according to a regimen that is prescribed by a physician. When the plant is dissipative and no feedback is involved, the resulting control system is simple and safe. However, the open-loop control cannot attenuate disturbances and handle plant uncertainty. 

Provided the actuators can be continuously manipulated and real-time measurements of the controlled variable 
are available, feedback control is routinely employed to achieve robust closed-loop stability or performance.
When the control signal is however restricted to \emph{impulsive} action, the only currently available feedback strategy is Model Predictive Control (MPC)~\cite{SPS15}. 
The utility and physiological coherence of impulsive MPC in drug delivery applications is readily recognized. A promising application of this control approach to insulin dosing in simulated diabetes patients is reported in e.g. \cite{RGS20}. In fact, impulsive insulin delivery mimics the physiological profile of secreting around ten major hormone pulses over 24 hours \cite{PGV88} with their temporal distribution  related to meals. Impulsive feedback control is inherently nonlinear and adding an advanced control law to the closed-loop dynamics further complicates stability and performance analysis. Yet, simple pulse-modulated feedback solutions manipulating the amplitude and frequency of the control impulses are lacking at present. 

The Impulsive Goodwin's Oscillator (IGO)  was proposed~\cite{MCS06,Aut09} as a hybrid (continuous-discrete) model of testosterone regulation in males, generalizing the concept of the original (continuous) Goodwin's oscillator~\cite{Good65} to the case of pulsatile (non-basal) secretion. 
The IGO possesses a number of properties that are typically sought for in biomedical applications, e.g. positivity and boundedness of the solutions. By design, the IGO has no equilibria and can only exhibit periodic or non-periodic (chaotic and quasiperiodic) oscillations \cite{ZCM12b}. It is proven that the IGO always possesses a unique (stable or unstable) $1$-cycle, i.e. a periodic solution with only one firing of the pulse-modulated feedback on the least period \cite{Aut09}.  Extensive bifurcation analysis of the IGO \cite{ZCM12b} suggests that the model, being equipped with the modulation functions of Hill's type, is monostable, even under a small delay present in the closed loop \cite{CMZh16}. Thus, when in a stable periodic solution, the IGO is not likely to change to another type of solution due to a temporary exogenous disturbance.

This paper addresses a novel problem of designing an IGO that exhibits a stable 1-cycle with desired characteristics. The main contributions of the paper are threefold:
\begin{itemize}
    \item the IGO design problem is formulated with respect to a desired solution, i.e. a 1-cycle;
    \item necessary and sufficient orbital stability conditions of the 1-cycle in the IGO are provided;
    \item bifurcation analysis of the nonlinear IGO  dynamics in vicinity of the designed 1-cycle is performed.
\end{itemize}

The paper is organized as follows. In Section~\ref{sec.known}, known facts about the dynamics of the IGO are summarized 
to facilitate further reading. In Section~\ref{sec.design},  the problem of designing an IGO that exhibits a stable predefined 1-cycle is formulated and solved. 
A numerical example is considered in Section~\ref{sec.numerical} to illustrate the proposed design concept.
Section~\ref{sec.bifur} provides bifurcation analysis of the designed IGO to discern nonlinear dynamics phenomena arising under deviations of the nominal parameter values. Finally,  conclusions are drawn.

\section{BACKGROUND}\label{sec.known}

This section summarizes the facts pertaining to  the IGO model and its behaviors that are used in the rest of the paper.

\subsection{The Impulsive Goodwin's Oscillator}
The IGO is given by  the following equations \cite{MCS06,Aut09}
\begin{equation}                            \label{eq:1}
\dot{x}(t) =Ax(t), \quad z(t)=Cx(t),
\end{equation}
\begin{equation}                             \label{eq:2}
\begin{aligned}
x(t_n^+) &= x(t_n^-) +\lambda_n B, \quad                                          
t_{n+1} =t_n+T_n,
\\ 
T_n &=\Phi(z(t_n)), \quad \lambda_n=F(z(t_n)),
\end{aligned}
\end{equation}
where $A,B,C$ are constant matrices, $n=0,1,\ldots$,
$$A=\left[\begin{smallmatrix} -a_1 &0 &0 \\ g_1 & -a_2 &0 \\ 0 &g_2 &-a_3 \end{smallmatrix}\right], B=\left[\begin{smallmatrix} 1 \\ 0 \\ 0\end{smallmatrix}\right], C =[0,0,1],
$$
$z$ is the controlled output, and the state
$x= [x_1,x_2,x_3]^\top$
describes concentrations  of some chemical substances.
In continuous model part~\eqref{eq:1}, $a_1,a_2,a_3>0$ are distinct constants and $g_1,g_2>0$ are positive gains.
It is readily observed that the matrix $A$ is Hurwitz stable, also, 
 \begin{equation}\label{CBLB}
 CB=0,\,CAB=0,\,CA^2B\ne 0.
 \end{equation}
The latter property implies, in particular, that $z(t)$ is a smooth function despite the jumps in \eqref{eq:2}.

The minus and plus in a superscript in~\eqref{eq:2} denote the left-sided and
a right-sided limit, respectively.
The amplitude modulation function $F(\cdot)$ and frequency modulation
function $\Phi(\cdot)$ are continuous and monotonic for positive arguments; $F(\cdot)$ is non-increasing
and $\Phi(\cdot)$ is non-decreasing, also,
\begin{equation}                             \label{eq:2a}
\Phi_1\le \Phi(\cdot)\le\Phi_2, \quad 0<F_1\le F(\cdot)\le F_2,
\end{equation}
where $\Phi_1$, $\Phi_2$, $F_1$, $F_2$ are positive constants\footnote{Notably, with respect to dosing applications, the bounds $F_1$ and $F_2$  specify the least and largest dose that can be delivered by the control law, while $\Phi_1$ and $\Phi_2$ prescribe the shortest and longest interval between the administered doses. The explicit way of enforcing these safety limits is favorable in, e.g., healthcare applications.}. Then control law~\eqref{eq:2} constitutes a frequency and amplitude
pulse modulation operator~\cite{GC98} implementing an output feedback over \eqref{eq:1}. The time instants $t_n$ are called (impulse) firing times
and $\lambda_n$ represent the corresponding impulse weights.

\subsection{Solution Properties}
The dynamics of the IGO are defined by differential equation \eqref{eq:1} in between the feedback firing times and undergo jumps of the magnitude $\lambda_nB$ at the times $t_n$ in accordance with~\eqref{eq:2}. Due to the positivity of $F_1$, the IGO lacks equilibria  and exhibits only  oscillatory periodic or non-periodic (e.g. chaotic or quasiperiodic) solutions.
The solutions of the IGO are positive under a positive initial condition $x(t_0^-)$, because $A$ is Metzler\footnote{A square matrix whose off-diagonal entries are all nonnegative is said to be Metzler.  The exponential $\e^{tA},t\geq 0$ is nonnegative for a Metzler $A$.} and $F(\cdot)$ is uniformly positive due to~\eqref{eq:2a}. It is proved in~\cite{Aut09} that the solutions are  bounded, because $A$ is Hurwitz and the nonlinear characteristics $F,\Phi$ are bounded.

Denoting  $X_n=x(t_n^-)$, the evolution of the continuous state vector of the IGO from one firing time to the next one obeys  the impulse-to-impulse map \cite{Aut09}
\begin{align}\label{eq:map}
    X_{n+1}&=Q(X_n),\\
    Q(\xi) &= \mathrm{e}^{A\Phi(C\xi)}\left( \xi+ F(C\xi)B \right).\nonumber
\end{align}

This paper focuses on periodic solutions of model \eqref{eq:1},\eqref{eq:2} that correspond to fixed points of the map $Q$. A periodic solution with exactly $m$ firings of the pulse-modulated feedback within the least period is called $m$-cycle. 
In particular, for a 1-cycle with the initial condition $X$, it applies
\begin{equation}\label{eq:1-cycle}
    X=Q(X).
\end{equation}
Since all the solutions of \eqref{eq:1}, \eqref{eq:2} are positive, it holds that $X>0$, where the inequality is understood element-wise. 
\begin{prop}[\cite{Aut09}]\label{th:1-cycle}
System~\eqref{eq:1},~\eqref{eq:2} has one and only one (positive) $1$-cycle, that is,~\eqref{eq:1-cycle} has
    a unique solution $X>0$. The cycle parameters $\lambda$, $T$, and $z_0$ can be evaluated by solving the following system of algebraic equations
    \begin{align}\label{eq.z0}
        z_0&=\lambda g_1 g_2 \sum_{i=1}^3 \frac{\alpha_i}{\mathrm{e}^{a_iT}-1}, \quad \alpha_i=\prod_{\substack{j=1\\j\ne i}}^3\frac{1}{a_j-a_i}, \\
        \lambda&= F(z_0), \quad T=\Phi(z_0).\label{eq.z0-2}
    \end{align}

The key idea of proving Theorem~\ref{th:1-cycle} in~\cite{Aut09} is to rewrite~\eqref{eq:1-cycle} in terms of the output variable $z=CX=x_3$  as
\begin{equation}\label{eq:1-cycle-equiv}
X=\mathrm{e}^{A\Phi(z)}\left(X + F(z)B \right),\;\;z=CX,
\end{equation}
which equation is subsequently reduced to the scalar equation
\begin{equation*}
z=C(\mathrm{e}^{-A\Phi(z)}-I)^{-1}BF(z).
\end{equation*}
The right-hand side of this equation is a decreasing bounded function of $z>0$, being strictly positive as $z\to 0+$~\cite{Aut09}, which entails the existence and uniqueness of the solution.

The  1-cycle above is orbitally asymptotically stable \cite{Aut09} if and only if the fixed point $X$ is asymptotically stable as the equilibrium of discrete-time dynamics~\eqref{eq:map}, that is, the Jacobian matrix $Q'(X)$ is Schur 
stable\footnote{A square matrix is said to be Schur (Schur stable) if all its eigenvalues $\lambda_j$ belong to the unit disk $|\lambda_j|<1$.}, where
\begin{equation}\label{eq:jacobian}
    Q^\prime(X)= \mathrm{e}^{A\Phi(z_0)}\left( I+ F^\prime(z_0)BC\right)+ \Phi^\prime(z_0)AX C.
\end{equation}
\end{prop}
\vskip0.1cm



\section{DESIGN}\label{sec.design}
The IGO \emph{design} problem treated here is formulated  in the following way. Suppose that the dynamics of  \eqref{eq:1} given by the matrix $A$ are known. In drug dosing, the elements of $x(t)$  can belong to, e.g., a known pharmacokinetic-pharmacodynamic model~\cite{Runvik:2020}.
Given the parameters of a 1-cycle, the IGO design task is to find the modulation functions that render, with orbital stability, the desired periodic solution.

In terms of the model parameters, the problem in question can be summarized as follows. Given the parameters $a_1,a_2,a_3,g_1$, find  $\Phi(\cdot), F(\cdot)$ that provide the desired characteristics of a stable 1-cycle $\lambda, T$. In the design procedure proposed below, $g_2>0$ always appears in product with $\lambda$ and can be selected as an arbitrary constant.

From \eqref{eq.z0-2} and \eqref{eq:jacobian}, the conditions for  1-cycle existence and stability in the IGO involve $z_0$, i.e. the output value at the fixed point $X$ in~\eqref{eq:1-cycle}. Therefore, the modulation functions, as such, cannot be obtained in the design procedure, but only interpolation conditions that they and their derivatives have to satisfy to achieve the desired solution.

\subsection{Divided differences and the Opitz formula}

To evaluate of a function $f(\cdot)$ of the matrix $A$, the so-called Opitz formula will be used in the analysis to follow. The complex-valued function $f(\cdot)$ is assumed to be well-defined and complex-analytic in a vicinity of the matrix spectrum $\sigma(A)=\{-a_1,-a_2,-a_3\}$ where the eigenvalues are pairwise different; See  \cite{PRM23} for a more general case. 

The first divided difference (1-DD) of a function $f$ is introduced \cite{D03,DeBoor2005} as a function of two variables 
\[
f[z_0,z_1]\triangleq \frac{f(z_1)-f(z_0)}{z_1-z_0},
\]
which expression is well defined if and only if $f(z_1)$, $f(z_0)$ exist and $z_0\ne z_1$. The second divided difference (2-DD) is a  function of three variables and is defined by
\[
f[z_0,z_1,z_2]\triangleq\frac{f[z_1,z_2]-f[z_0,z_1]}{z_2-z_0},
\]
where $f(z_0),f(z_1),f(z_2)$ exist and $z_0,z_1,z_2$ are pairwise different.
Remarkably, both 1-DD and 2-DD are \emph{symmetric} functions. Furthermore, for a scalar $\xi\ne 0$ and $f_{\xi}(z)\triangleq f(z\xi)$, it holds
\[
f_{\xi}[z_0,z_1]=\xi f[z_0\xi,z_1\xi],\;f_{\xi}[z_0,z_1,z_2]=\xi^2f[z_0\xi,z_1\xi,z_2\xi].
\]
After some computations, it can be shown that
\[
f[z_0,z_1,z_2]=\sum_{i=0}^2\beta_if(z_i),\quad
\beta_i=\prod_{\substack{j=1\\j\ne i}}^3\frac{1}{z_j-z_i}.
\]

The Lagrange mean value theorem implies that if $f(\cdot)$ attains \emph{real} values on some real interval $I=(\alpha,\beta)$, then, for each $z_0,z_1\in I$, $z_0<z_1$
there exists $\zeta\in[z_0,z_1]$ such that
$
f[z_0,z_1]=f'(\zeta).
$
A similar result can be proved for the 2-DD~\cite[Corollary to Proposition~43]{DeBoor2005}: for each triple $z_0,z_1,z_2\in I$, one has
\[
f[z_0,z_1,z_2]=\frac{1}{2}f''(\zeta),\;\;\zeta\in[\min_i z_i,\max_i z_i].
\]

For matrices of dimension three, a generalized\footnote{Typically, the Opitz formula is considered for the situation where the second main diagonal contains ones, that is, $g_1=g_2=1$, the general case can be derived by a simple similarity transformation.} Opitz formula in \cite{PRM23} gives the closed-form representation of $f(A)$
\begin{equation*}
f(A)=
\begin{bmatrix}
f(-a_1) & 0 & 0\\
g_1f[-a_1,-a_2] & f(-a_2) & 0\\
g_1g_2f[-a_1,-a_2,-a_3] & g_2f[-a_2,-a_3] & 0
\end{bmatrix}.
\end{equation*}

For instance, one may compute explicitly the evolutionary matrix of the linear system \eqref{eq:1} by applying the Opitz formula to
function $f(z)=\exp(zt)$, where $t\in\mathbb{R}$ is constant:
\begin{align*}
&\exp(At)=\\
&\left[\begin{array}{c:c:c}
    \e^{-a_1t} &0 &0\\
    g_1t \e \lbrack -a_1t, -a_2t \rbrack &\e^{-a_2t} &0\\
    g_1g_2t^2 \e \lbrack -a_1t, -a_2t,  -a_3t\rbrack & g_2t\e \lbrack -a_2t, -a_3t \rbrack & \e^{-a_3t}
\end{array}\right].
\end{align*}
Here, following standard notation, we use $\e[z_0,z_1]$ to denote the 1-DD of the exponential function $\e^z=\exp(z)$; the same applies to the 2-DD $\e[z_0,z_1,z_2]$.

By virtue of the mean value theorem, all divided differences of the exponential function are positive. Subsequently, all the elements of $\exp At$ are non-negative.  This is well in line with the fact of $A$ being Metzler. The obtained expression for the transition matrix generalizes to higher dimensions of the continuous dynamics when the two-diagonal structure of the matrix $A$ is preserved~\cite{PRM23}.




\subsection{Fixed point}

Proposition~\ref{th:1-cycle}, combined with the Opitz formula, enables the calculation of the  parameters of the unique 1-cycle for a given model of the IGO. The 1-cycle corresponds to a fixed point of the map $Q(\cdot)$, according to \eqref{eq:1-cycle}. 
The following converse statement, 
yielding the fixed point for a set of 1-cycle parameters,  can then be proven. 

Denote, for brevity,
\[
\mu(z)\triangleq \frac{1}{\e^{-z}-1}=\frac{\e^z}{1-\e^z}, \quad z\ne 0.
\]

\begin{prop}\label{th:fp}
    Given the parameters of $1$-cycle $T>0$, $\lambda>0$, the fixed point $X>0$ of map $Q$ from~\eqref{eq:map} is calculated as
\begin{align}\label{eq:x0}
     x_{1}&=\lambda\mu(-a_1T)=\frac{\lambda\mathrm{e}^{-a_1T}}{1-\mathrm{e}^{-a_1T}}, \nonumber\\
     x_{2}&=\lambda g_1T\mu[-a_1T,-a_2T]=\nonumber\\
     &=\frac{\lambda g_1 T\e\lbrack -a_1T,-a_2T \rbrack}{(1-\mathrm{e}^{-a_1T})(1-\mathrm{e}^{-a_2T})}, \\
     x_{3}&=\lambda g_1g_2T^2\mu[-a_1T,-a_2T,-a_3T]=\nonumber\\
     &=
           \frac{\lambda g_1 g_2 T^2}{(1-\mathrm{e}^{-a_1T})(1-\mathrm{e}^{-a_2T})(1-\mathrm{e}^{-a_3T})}\times\\
           &\times\Big(\e\lbrack -a_1T,-a_2T, -a_3T \rbrack\nonumber \\
           &+\e\lbrack -(a_1+a_2)T,-(a_1+a_3)T, -(a_2+a_3)T \rbrack\Big).\nonumber
 \end{align}
\end{prop}

\begin{proof}
For $\Phi(CX)=\Phi(x_{03})=T$ and $F(CX)=F(x_{03})=\lambda$, $X$ is a given fixed point satisfying \eqref{eq:1-cycle-equiv} if and only if
\[
X=\lambda(\e^{-AT}-I)^{-1}B=\lambda\mu(AT)B,
\]
that is, $X$ is the first column of the matrix $\mu(AT)$.  The leftmost equalities in~\eqref{eq:x0}, relating $x_{0i}, i=1,2,3$ to the divided
differences of $\mu$, follow immediately from the Opitz formula. The 
rightmost equalities are validated by a straightforward computation, which is omitted here.
\end{proof}


 Proposition~\ref{th:fp} implies that $z_0=x_{03}$ can be calculated for any choice of the distinct constants $a_1,a_2,a_3$, which fact perfectly agrees with the result of Proposition~\ref{th:1-cycle}. Then, for a given continuous part of the IGO in \eqref{eq:1} and desired $\lambda, T$, the value of $z_0$ is obtained by specifying the values of the modulation functions at that point according to \eqref{eq.z0-2}. Further, since a 1-cycle is uniquely defined by the fixed point, the elements of the matrix $A$ and $\lambda, T$ stipulate the periodic solution of the IGO. 
 

\subsection{Stability of $1$-cycle}
Proposition~\ref{th:fp} specifies the fixed point corresponding to the desired periodic solution but does not guarantee its stability.  Then, additionally, matrix~\eqref{eq:jacobian} needs to be stable to ensure that the 1-cycle is relevant in feedback control context.

In the design problem at hand, the slopes of the modulation functions $F(\cdot)$, $\Phi(\cdot)$ at the fixed point corresponding to the desired 1-cycle constitute the degrees of freedom that can be utilized for the stabilization of the periodic solution. As the result below explicates,  the design problem is similar to what is known as static output feedback stabilization in linear time-invariant (LTI) systems \cite{SAD97}.

\begin{prop}\label{th:feedback} Jacobian \eqref{eq:jacobian} at the fixed point $X$ admits the parameterization
    \begin{equation*}
        Q^\prime(X)= \e^{A\Phi(z_0)}+ \left( F^\prime(z_0)J+\Phi^\prime(z_0)D\right)C,
    \end{equation*}
    where  $J,D\in \mathbb{R}^3$ and $J=\e^{A\Phi(z_0)}B>0$, $D=AX<0$, $z_0=CX=x_{03}$.
\end{prop}

\begin{proof}
The expression for $Q'(X_0)$ and formulas for $J,D$ are straightforward from~\eqref{eq:jacobian}. Furthermore, since $g_1,g_2>0$ and all divided differences of the exponential functions are positive, the 
formula for $\e^{At}$ derived in Section~III-A ensures that the vector $J$, being the first column of the matrix $\e^{A\Phi(z_0)}$, is strictly positive. In order to prove that $D=AX<0$, notice that
\[
D=A(\e^{-A\Phi(z_0)}-I)^{-1}B.
\] 
Introducing the function
\[
\nu(z)\triangleq z\mu(z)=\frac{z}{e^{-z}-1},
\]
one notices that $D=T^{-1}\nu(TA)B$ is nothing else but the first column of the matrix $T^{-1}\nu(TA)$. It can be demonstrated that the function $\nu$ (see Fig.~\ref{fig.nu}) is negative, decreasing, and strictly concave on the interval $z\in[-\infty,0)$. Hence, in view of the mean value theorem, the divided differences $\nu[-a_1T,-a_2T]$, $\nu[-a_2T,-a_3T]$, $\nu[-a_1T,-a_2T,-a_3T]$ are all negative, as well the values $\nu(-a_iT)$.
In virtue of the Opitz formula, $\nu(TA)B<0$, entailing 
that $D<0$ and concluding the proof of Proposition~\ref{th:feedback}.
\end{proof}
\begin{figure}[htb]
\includegraphics[width=\columnwidth]{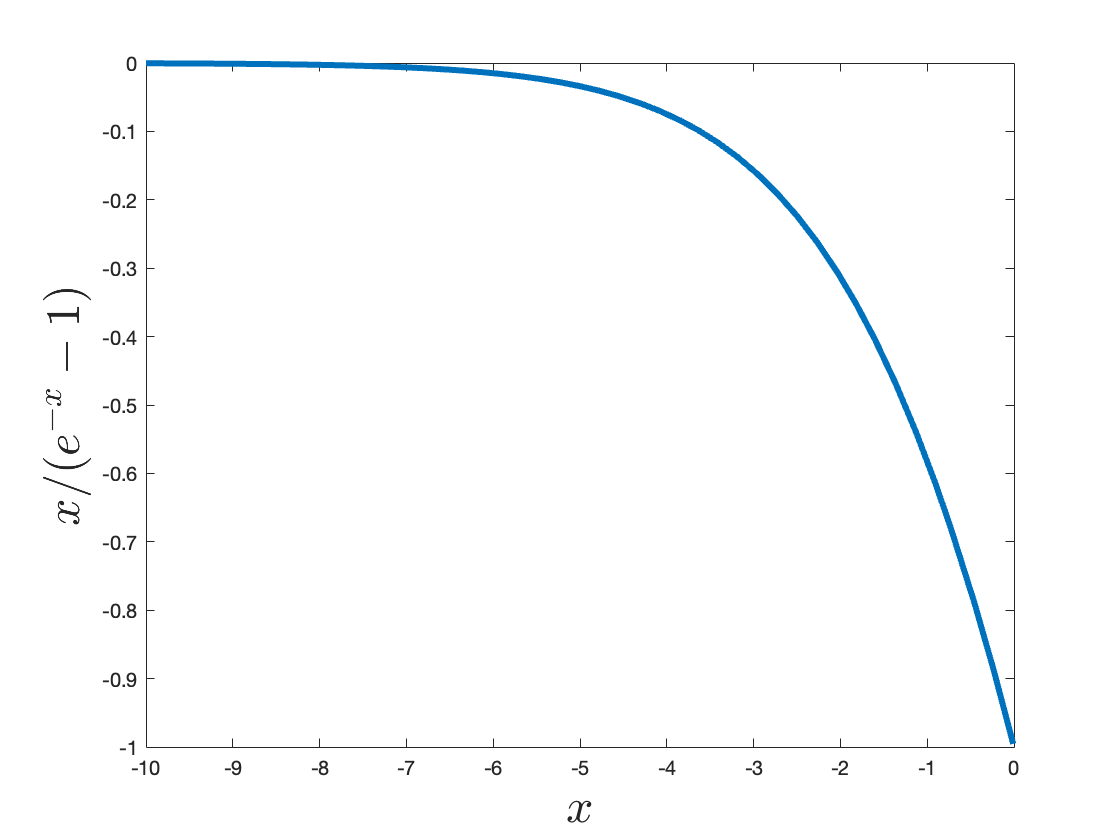}
\caption{The plot of function $\nu(x)$ for $x<0$.}\label{fig.nu}
\end{figure}


From the result of Proposition~\ref{th:feedback}, $Q^\prime(X)$ can be rendered Schur stable by the feedback gain $K\in \mathbb{R}^3$
\begin{equation}\label{eq:closed_loop}
    Q^\prime(X)= \e^{A\Phi(z_0)}+KC,
\end{equation}
subject to
\begin{equation}\label{eq:gain}
    K= \begin{bmatrix}
    J &D
\end{bmatrix}\begin{bmatrix}
 F^\prime(z_0)  \\  \Phi^\prime(z_0) 
\end{bmatrix}.
\end{equation}
Since the pair $(\e^{A\Phi(z_0)}, C)$ is observable, an arbitrary eigenvalue spectrum of $Q^\prime(X)$ can be achieved with an unrestricted gain $K$.  However, due to \eqref{eq:gain}, $K$ has to be a linear combination of $J$ and $D$ with the coefficients  $F^\prime(z_0)\le 0$ and $ \Phi^\prime(z_0)\ge 0$, correspondingly. This feedback structure also appears in the classical problem of static output feedback design, see \cite{SAD97} for an overview.
A crucial distinction between the static output feedback in an LTI system and the pulse-modulated feedback of the IGO is that the former operates around a (constant) output setpoint whereas the latter stabilizes an LTI along a periodic solution (a 1-cycle) expressed as a  fixed point.

\begin{remark} The last statement of Proposition~\ref{th:feedback} entails that
$JF^\prime(z_0) + D\Phi^\prime(z_0)\leq 0$, for all feasible values of $F^\prime(z_0),\Phi^\prime(z_0)$. Therefore, the feedback in the IGO is {\it negative}, despite the fact that all the involved quantities are {\it positive}.
This property is natural given the underlying principle of the pulse-modulated feedback in the IGO where the impulses become of lower weight and sparser when the output values are higher than $z_0$.
\end{remark}

It can also be noticed that the pair of slopes $F^\prime(z_0)=0,\Phi^\prime(z_0)=0$ yields in the Schur stable matrix $Q'(X)$. Even though constant modulation functions formally produce  a stable 1-cycle, the feedback in the IGO is essentially eliminated, and the impulsive sequence is independent of the measured output.

\begin{lem}[Theorem~3.1, \cite{FGL98}]\label{lm:Schur}
Let $A=\lbrack a_{ij}\rbrack_{i,j=1}^3$ be a real matrix. Denote $M(A)=m_{11}(A)+m_{22}(A)+m_{33}(A)$, where $m_{ii}(A)$ stand for
the principle minors
\begin{align*}
    m_{11}(A)&= a_{22}a_{33}-a_{23}a_{32},\\
     m_{22}(A)&= a_{11}a_{33}-a_{31}a_{13}, \\
      m_{33}(A)&= a_{11}a_{22}-a_{21}a_{12}.
\end{align*}
Then, matrix $A$ is Schur stable if and only if the following three conditions are satisfied:
\begin{enumerate}
    \item $| \det A |< 1$,
    \item $ | \tr A+ \det A |<1+ M(A) $,
    \item $ |\tr A\det A - M(A)|< 1- \det^2 A$.
\end{enumerate}
\end{lem}
\vskip0.2cm

To analyse the Schur stability of matrix~\eqref{eq:jacobian}, one can find the characteristics employed by Lemma~\ref{lm:Schur} as functions of
$\Phi^\prime(z_0)$, $F^\prime(z_0)$. 
For instance, applying~\eqref{eq:closed_loop} and the well-known Schur complement formula $\det(I+XY)=\det(I+YX)$, where $XY,YX$ are square matrices, but $X,Y$ need not be square, one has
\[
\begin{split}
\det Q'(X)=\det(\e^{AT})\det(I_3+\e^{-AT}KC)=\\=\e^{-(a_1+a_2+a_3)T}(1+C\e^{-AT}K)=\\=\e^{-(a_1+a_2+a_3)T}\left(1+C\e^{-AT}[J,D]\left[\begin{smallmatrix}
F'(z_0)\\\Phi'(z_0)\end{smallmatrix}\right]\right)=\\
\e^{-(a_1+a_2+a_3)T}(1+Ce^{-AT}D\Phi'(z_0)).
\end{split}
\]
To derive the latter equality, one has to notice that $C\e^{-AT}J=C\e^{-AT}\e^{AT}B=CB=0$.
Similarly, after some computations, one can obtain two remaining characteristics. We formulate 
the following proposition. 

\begin{prop}\label{th:Jacobian}
    For $Q^\prime(X)$ defined by \eqref{eq:jacobian}, it applies
    \begin{align*}
        \tr Q^\prime(X)&= \tr\e^{AT}+ C \begin{bmatrix}
            \e^{AT}B & AX\end{bmatrix} 
            \begin{bmatrix} F^\prime(z_0) \\ \Phi^\prime(z_0) \end{bmatrix}, \\
        \det Q^\prime(X)&= \e^{-(a_1+a_2+a_3)}T(1+C\e^{-AT}AX\Phi'(z_0)),\\
        M(Q^\prime(X))&= \e^{-(a_1+a_2)T}+\e^{-(a_1+a_3)T}+\e^{-(a_2+a_3)T}\\
        &+ \begin{bmatrix} \psi_1 & \psi_2\end{bmatrix}\begin{bmatrix} F^\prime(z_0) \\ \Phi^\prime(z_0) \end{bmatrix}, \\
       \psi_1 &= ( \e^{-a_1T}+\e^{-a_2T})j_3\\
              &- g_2T\Big(\e\lbrack -a_2T,-a_3T\rbrack j_2\\
              &+ g_1T \e \lbrack -a_1T,-a_2T,-a_3T\rbrack j_1\Big),\\
     \psi_2 &= ( \e^{-a_1T}+\e^{-a_2T})d_3\\
            &- g_2T\Big(\e\lbrack -a_2T,-a_3T\rbrack d_2\\
              &+ g_1T \e \lbrack -a_1T,-a_2T,-a_3T\rbrack d_1\Big).
    \end{align*}
    Here $j_i,d_i$ are the elements of vectors $J$ and $D$.
\end{prop}

\subsection{Design algorithm}\label{sec.algorithm}
The results of Section~\ref{sec.design} can be summarized in the form of the following procedure rendering the desired solution to the IGO.
\begin{enumerate}
    \item[Step 1:]\label{itm:S1} \  Select the desired 1-cycle's characteristics $\lambda$ and $T$.
    \item[Step 2:]\label{itm:S2} \ From plant model \eqref{eq:1}, obtain the parameters $a_1$, $a_2$, $a_3$ and $g_1$;  $g_2>0$ can be selected arbitrarily.
    \item[Step 3:]\label{itm:S3} \  Calculate the fixed point (and $z_0$) from \eqref{eq:x0}.
    \item[Step 4:]\label{itm:S4} \ Define the structure of the  modulation functions $F$, $\Phi$ and calculate their derivatives $F^\prime$, $\Phi^\prime$.
    \item[Step 5:]\label{itm:S5} \ Evaluate the three stability conditions specified in Lemma~\ref{lm:Schur} with respect to the Jacobian $Q^\prime(X)$ using the expressions of the matrix functions in Proposition~\ref{th:Jacobian}.
     \item[Step 6:]\label{itm:S6} \ By selecting the parameters of   the modulation functions, ensure that $F^\prime(z_0)$, $\Phi^\prime(z_0)$ satisfy the stability conditions of Step~5.
    \item[Step 7:]\label{itm:S7} \ By  scaling  the modulation functions, ensure the equalities $F(z_0)=\lambda$, $\Phi(z_0)=T$.
\end{enumerate}
\section{DESIGN EXAMPLE}\label{sec.numerical}
This section illustrates the use of the design algorithm  outlined in Section~\ref{sec.algorithm} by a numerical example worked out step-by-step.


Consider the  design of a 1-cycle with $\lambda=4.66$, $T=66.75$ (Step~1) in the IGO given by \eqref{eq:1}, \eqref{eq:2}, where $g_1=2.0$, $g_2=0.5s$, {$a_1 =0.08$},  $a_2 =0.15$, $p=2$, and
$a_3=0.12$, (Step~2). The corresponding fixed point (Step~3)
\[
X=\begin{bmatrix}
    0.0225 &0.6360 &6.8330
\end{bmatrix}^{\top},
\]
thus $z_0=6.8330$.
Following \cite{Aut09}, define the structure of the modulation functions (Step~4) as the Hill functions
\begin{gather}\label{eq:fi}
\Phi(z) = k_1+k_2\:\dfrac{(z/h_\Phi)^{p_\Phi}}{1+(z/h_\Phi)^{p_\Phi}},\\
F(z)=k_3+\dfrac{k_4}{1+(z/h_F)^{p_F}}.\notag
\end{gather}
The coefficients $k_i, i=1,\dots,4$ explicitly specify the values on the minimal and maximal dose as well as the minimal and maximal time interval between the doses
\[
  k_3 <F(z)< k_3+k_4, \quad k_1 <\Phi(z)< k_1+k_2.
\]
The parameters $k_2$ and $k_4$  also influence the derivatives of the modulation functions
\begin{gather*}
\Phi'(z)=\dfrac{k_2 p_\Phi z^{p_\Phi-1} h_\Phi^{-p_\Phi}}{{(1+(z/h_\Phi)^{p_\Phi})}^2},\;
F'(z)=-\dfrac{k_4 p_F z^{p_F-1} h^{-p_F}}{(1+(z/h_F)^{p_F})^2}.
\end{gather*}
Therefore, besides $k_2$, $k_4$, the derivatives are also defined by $h_\Phi$, $p_\Phi$, $h_F$, $p_F$. 

From the parameters of  continuous part \eqref{eq:1} and $T$, the stability conditions of the fixed point $X$ are evaluated. Then the involved functions of the Jacobian amount to
\begin{align*}
    \tr Q^\prime(X)&= 0.0052 +  1.4574 F'(z_0)   -0.5020\Phi'(z_0),\\
    \det Q^\prime(X)&= 7.1410\cdot 10^{-11}   -0.172\cdot10^{-14}\Phi'(z_0),\\
    M(Q^\prime(X))&= 2.1528\cdot 10^{-7} -0.1251\cdot 10^{-4} F'(z_0)\\
    &+ 0.1460\cdot 10^{-4} \Phi'(z_0).
\end{align*}
Notice that $M(Q^\prime(X))>0$ for all admissible values of $F'(z_0)$, $\Phi'(z_0)$.
Given the orders of the coefficients in the matrix functions of $Q^\prime(X)$, stability of $X$ is guaranteed if 
\begin{align*}
        |\tr Q^\prime(X)|&< 1+ M(Q^\prime(X)),\\
    M(Q^\prime(X))&<1,
\end{align*}
or, due to the positivity of $M(Q^\prime(X))$,
\begin{equation}\label{eq:bound}
|\tr Q^\prime(X)|< 1.
\end{equation}
The inequality above is satisfied for $F'(z_0)= -0.1143$, $\Phi'(z_0)= 2.2852$. 
As expected, stability condition \eqref{eq:bound} limits the derivatives of the modulation functions that act as feedback gains, cf. \eqref{eq:gain}.
Also  $h$, $p$ have to obey certain inequalities imposed by the parametrization in \eqref{eq:fi}.

Introduce the notation
\[
\eta_\Phi={\left(\frac{z_0}{h_\Phi}\right)}^{p_\Phi}, \quad \theta_\Phi=\frac{k_2p_\Phi}{2z_0\Phi^\prime(z_0)}.
\]
Then, for  $\Phi'(z)$ to take the desired value in $z_0$, it applies 
\[
\eta_\Phi^2+ 2( 1-\theta_\Phi)\eta_\Phi+1=0,
\]
and, therefore, 
\[
\eta_{\Phi,{1,2}}= \theta_\Phi-1\pm\sqrt{\theta_\Phi(\theta_\Phi-2)}.
\]
When it is guaranteed that
\begin{equation}\label{eq:p_Phi}
    p_\Phi > \frac{4z_0\Phi^\prime(z_0)}{k_2}>0,
\end{equation}
both $\eta_{\Phi,1}$ and $\eta_{\Phi,2}$ are positive, then
\[
h_{\Phi,1,2}=\frac{z_0}{\sqrt[{p_\Phi}]{\eta_{\Phi,1,2}}}.
\]

Similarly, with 
\[
\eta_F={\left(\frac{z_0}{h_F}\right)}^{p_F}, \quad \theta_F=\frac{k_4p_F}{2z_0 F^\prime(z_0)},
\]
one has
\[
\eta_F^2+ 2( 1+\theta_F)\eta_F+1=0,
\]
and then
\[
\eta_{F,1,2}= -(\theta_F+1)\pm\sqrt{\theta_F(\theta_F+2)}.
\]
When it is guaranteed that
\begin{equation}\label{eq:p_F}
    0<p_F<-\frac{4z_0F^\prime(z_0)}{k_4},
\end{equation}
both roots are positive. Notice that the condition
\[
\theta_F+1<0
\]
results in a weaker inequality
\[
0<p_F<-\frac{2z_0F^\prime(z_0)}{k_4}.
\]
Conditions \eqref{eq:p_Phi} and \eqref{eq:p_F} are satisfied (Step~5) for $p_\Phi=p_F=2$, thus yielding $h_\Phi=h_F=h=4.112$, $k_2 =40$, $k_4=2.0$. Now, $k_1=60$ and  $k_3 = 3.0$ ensures (Step~6) that
\[
F(z_0)=\lambda, \quad \Phi(z_0)=T.
\]
The closed orbit of the designed 1-cycle is depicted in Fig.~\ref{f81}, along with trajectories resulting from deviations in initial conditions for continuous part of the IGO \eqref{eq:1}. The evolution of the impulse weight (dose) sequence $\lambda_k$ (see \eqref{eq:2}) to the pre-defined 1-cycle amplitude $\lambda$ is depicted in Fig.~\ref{f82}. A series of interchanging overdosing and underdosing events asymptotically converges to the desired value. This behavior could not be predicted from the design procedure since only a stable 1-cycle is sought. 

\begin{figure}[ht]
\centering \centering
\includegraphics[width=0.5\linewidth]{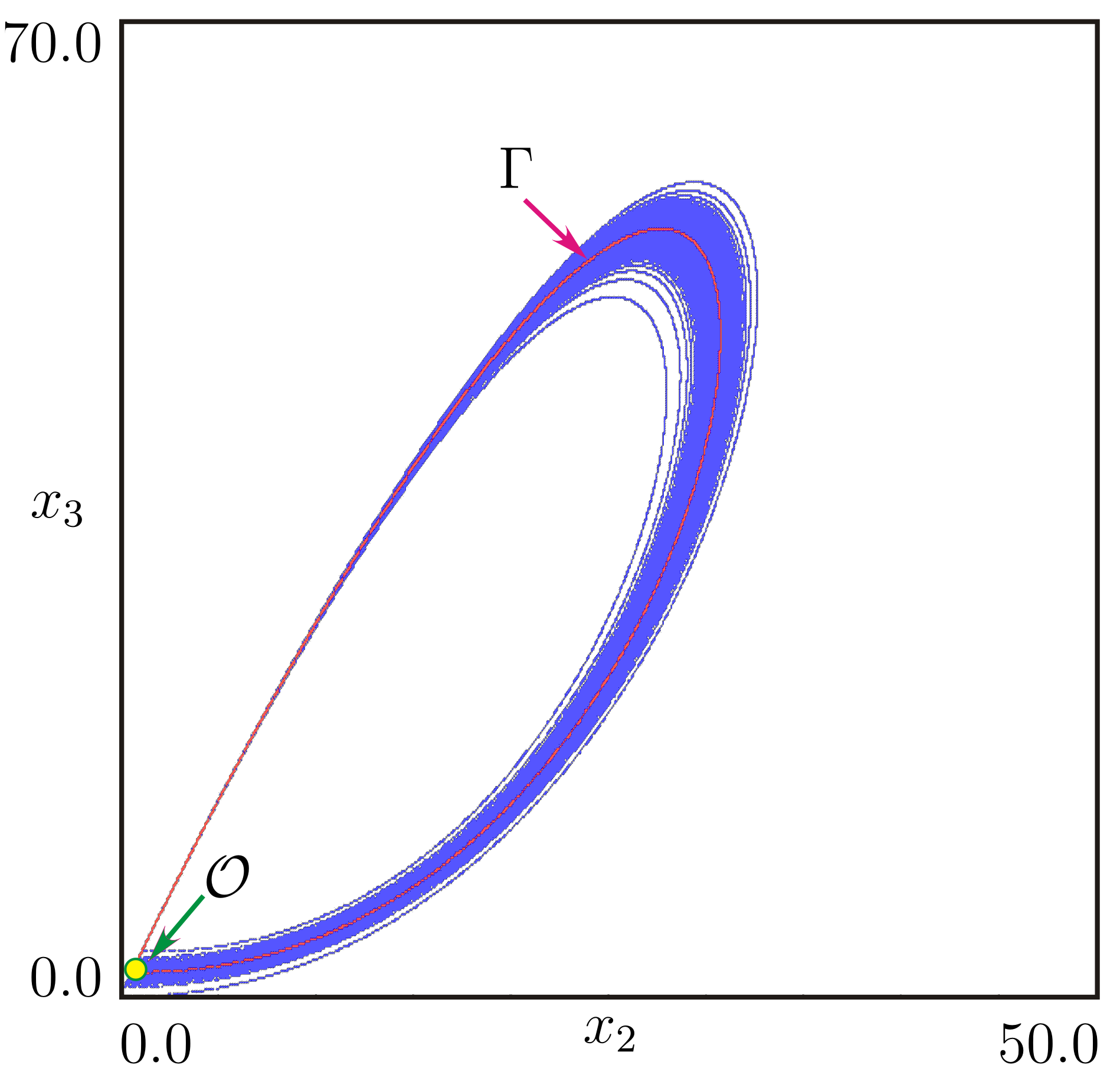}
 \caption{\label{f81}\small The designed 1-cycle ($\Gamma$, in red)  corresponding to the fixed point ($\mathcal{O}=X$). 
 Trajectories converging to $\Gamma$ are in blue.
 }
\end{figure}

\begin{figure}[ht]
\centering \centering
\includegraphics[width=0.6\linewidth]{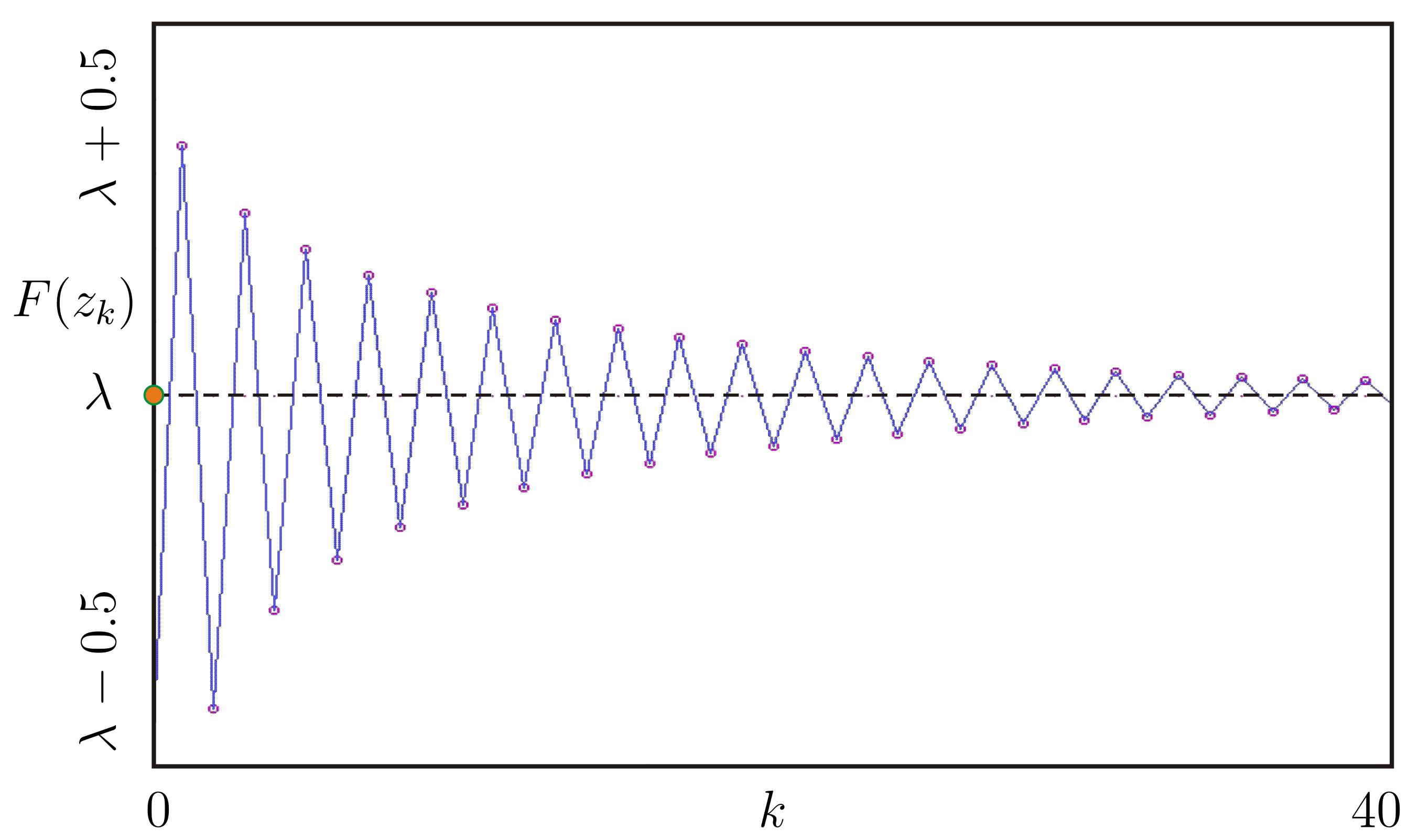}
 \caption{\label{f82}\small The convergence of the sequence $F(z_k)$ to  the desired $\lambda$.
 Since all the multipliers are negative $-1<\rho_i<0$, $1\leqslant i\leqslant 3$,  the convergence 
 is non-monotonous. To highlight the evolution, the point
 $F(z_{k-1})$  is connected to the next one $F(z_{k})$ (blue lines).
 }
\end{figure}

\section{BIFURCATION ANALYSIS}\label{sec.bifur}
To investigate the behaviors of the designed IGO  under parameter variation, bifurcation analysis is performed. 
For an interval of the parameter values $a_3$ and following the steps of the design procedure in Section~\ref{sec.design}, the value of $h=h_\Phi=h_F$ is found  and  the stability of a fixed point  $\mathcal{O}(a_3,h)=X$ is evaluated. The condition $h_\Phi=h_F$ is imposed to reduce the number of independent bifurcation parameters.

From Section~\ref{sec.algorithm},  $k_1=60$; $k_2 =40$; $k_3 = 3.0$; $k_4=2.0$;
$g_1=2.0$; $g_2=0.5$.
\textcolor{black}{ For each  $a_3$, the value of $h$ is found
by solving equations \eqref{eq:fi} with $F(z_0)=\lambda$, $\Phi(z_0)=T$ and
 the stability of a fixed point
$\mathcal{O}(a_3,h)=X$} of mapping~\eqref{eq:map} given by  Proposition~\ref{th:fp} is analyzed. 
\begin{figure*}[h]
        \begin{subfigure}[t]{0.24\textwidth}
            \centering
            \includegraphics[width=0.9\columnwidth]{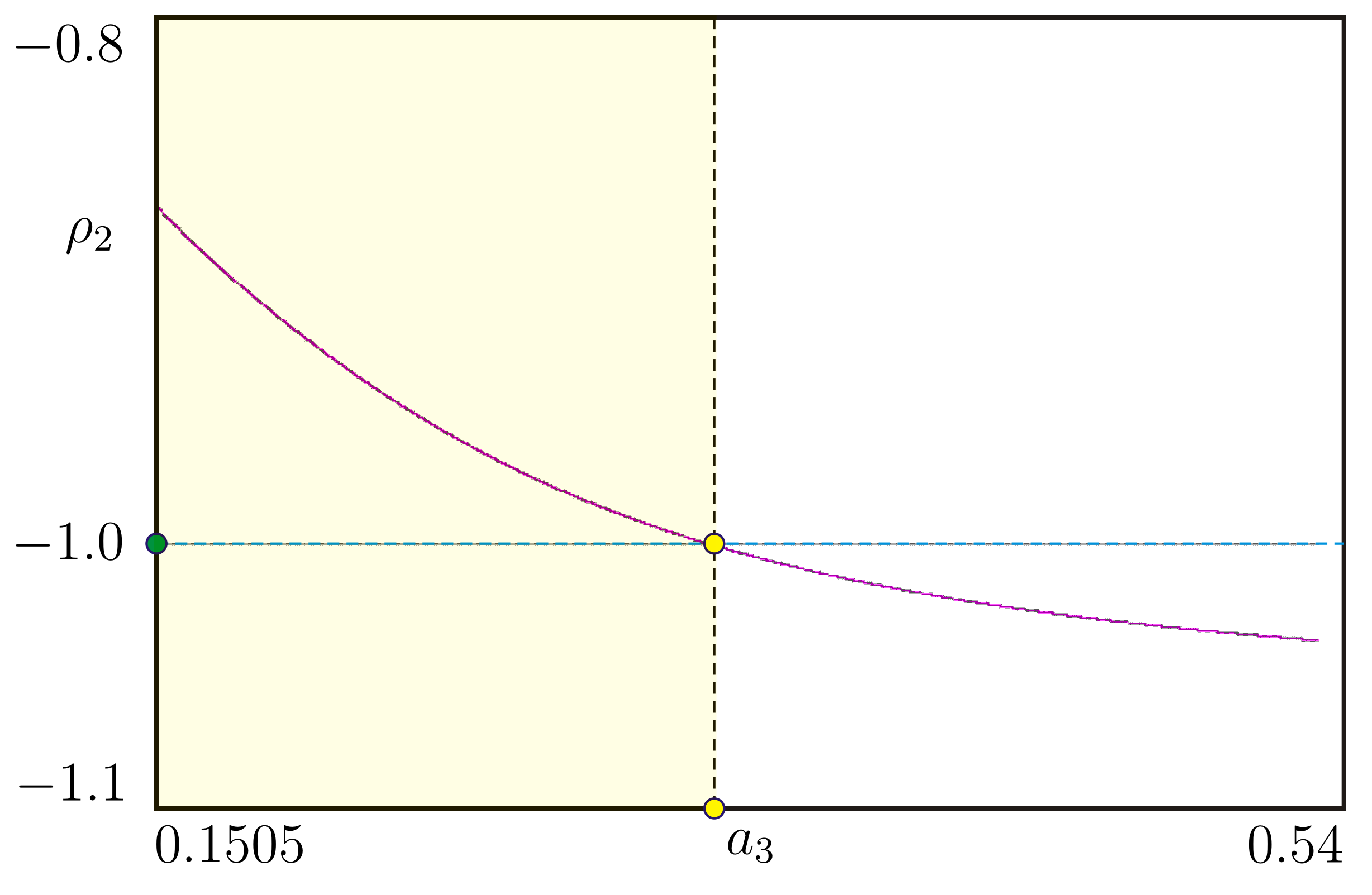}
            \caption[Network2]%
            {{\small Variation  of the maximal in absolute
value multiplier $\rho_2$ of a fixed point; $0.1505<a_3<0.54$. }}    
            \label{f69a}
        \end{subfigure}
        \hfill
        \begin{subfigure}[t]{0.24\textwidth}  
            \centering 
            \includegraphics[width=0.9\columnwidth]{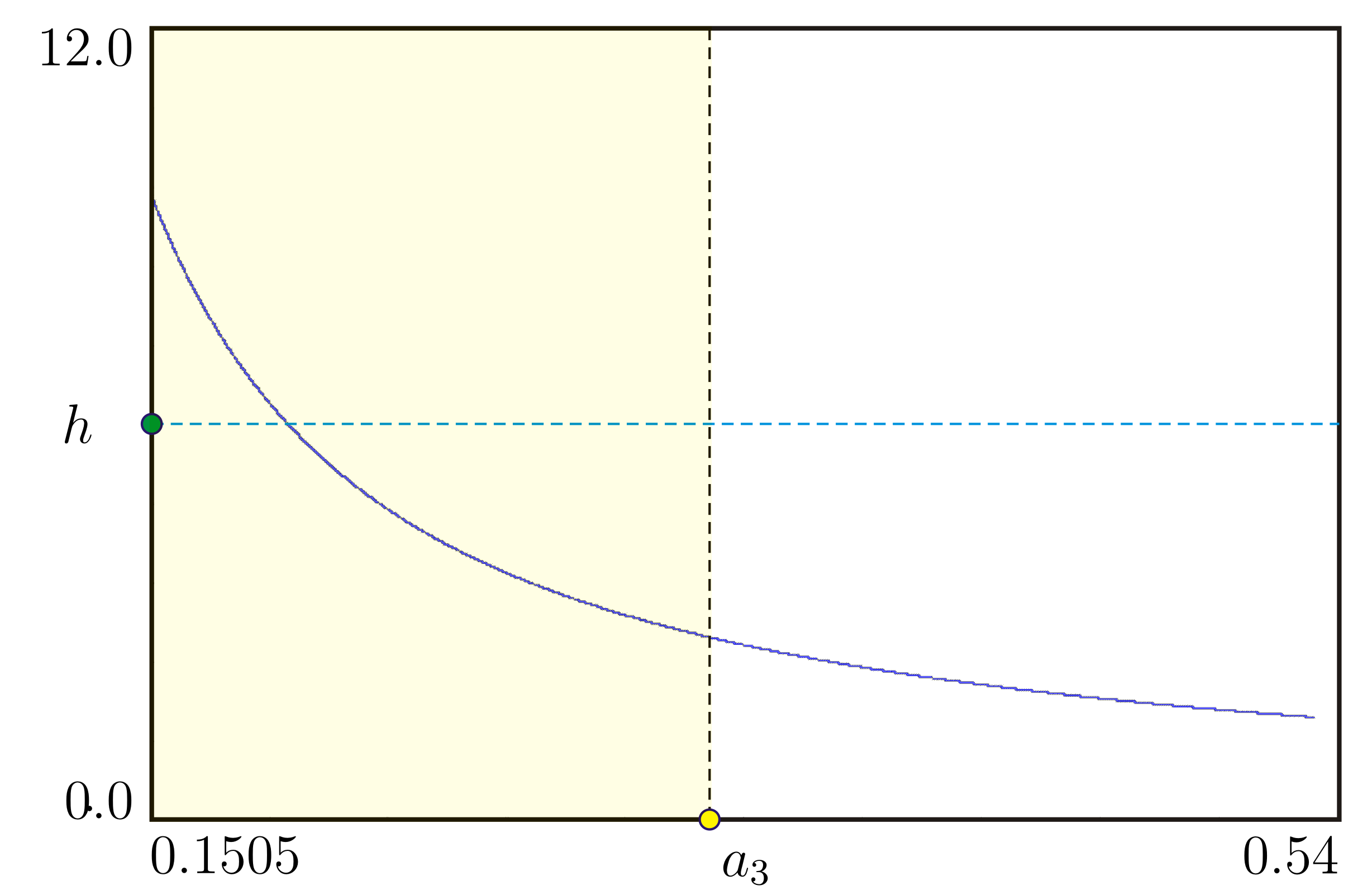}
            \caption[]%
            {{\small Dependence of $h$ on $a_3$}}    
            \label{f69b}
        \end{subfigure}
        \hfill
        \begin{subfigure}[t]{0.24\textwidth}   
            \centering 
            \includegraphics[width=0.9\columnwidth]{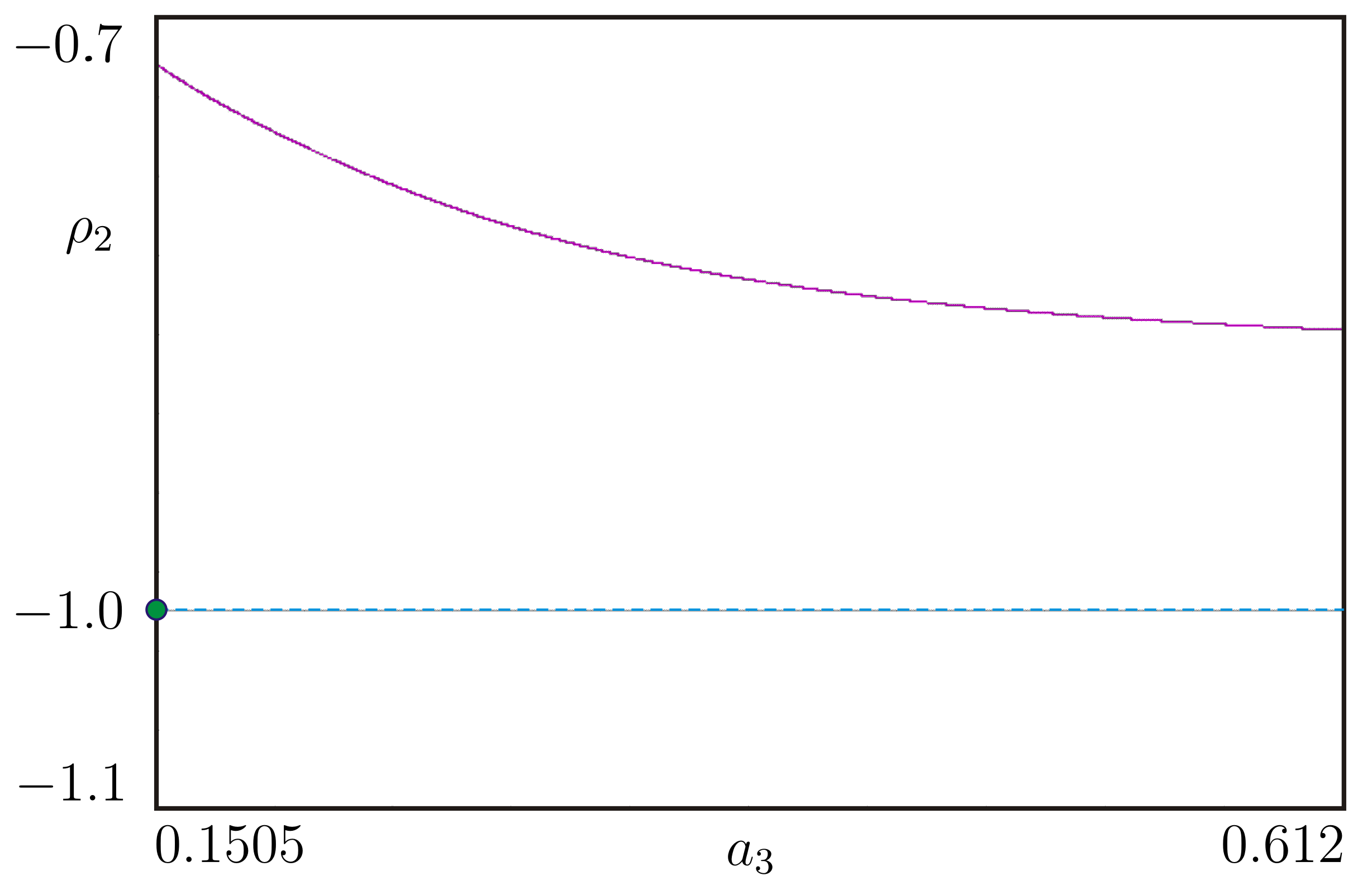}
            \caption[]%
            {{\small Variation  of the maximal in absolute value multiplier $\rho_2$; $0.1505<a_3<0.612$.}}    
            \label{f69c}
        \end{subfigure}
        \hfill
        \begin{subfigure}[t]{0.24\textwidth}   
            \centering 
            \includegraphics[width=0.9\columnwidth]{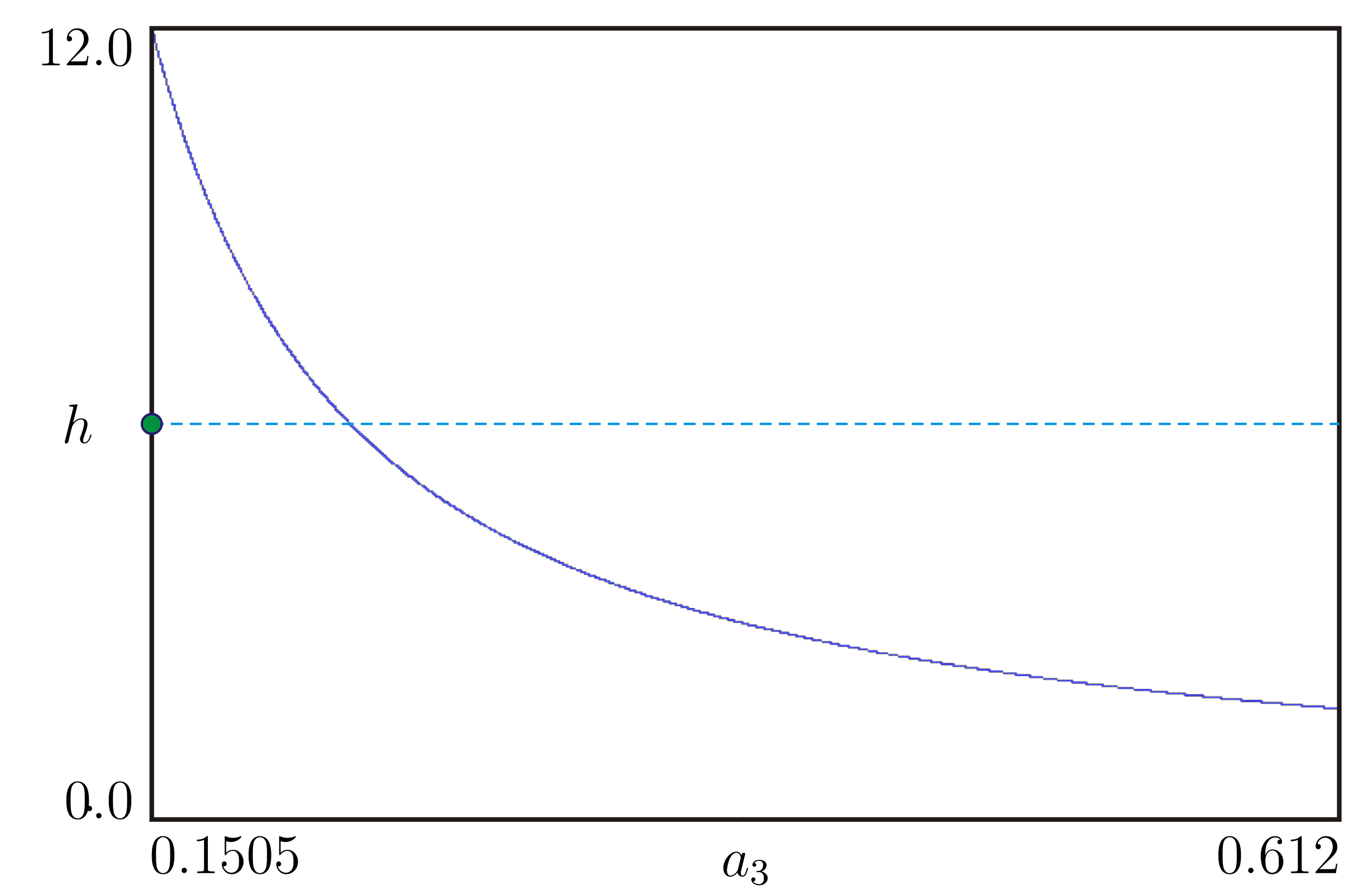}
            \caption[]%
            {{\small Dependence of $h$ on $a_3$}}    
            \label{f69d}
        \end{subfigure}
        \caption[ The average and standard deviation of critical parameters ]
        {\small Bifurcation analysis: (a),(b) - for $T=66.7502$, $\lambda = 4.66$, (c),(d) - for $T=65.4542$, $\lambda=4.7273$.} 
        \label{f69}
    \end{figure*}

\textcolor{black}{An example of such
an analysis is shown in Fig.~\ref{f69} (a),(b) for $T = 66.75$, $\lambda = 4.66$, and
$0.1505<a_3<0.54$. When the parameter $a_3$
increases, the fixed point $\mathcal{O}=X$ undergoes
a period-doubling bifurcation: the maximal in absolute
value multiplier  $\rho_2$ of the fixed point $\mathcal{O}$ emerges from the
unit circle though $-1$ (see Fig.~\ref{f69}(b)). In these figures, the stability region 
of the fixed point $\mathcal{O}$ is in yellow. Fig.~\ref{f69}(b) depicts the
dependence of $h$ on $a_3$ in the transition shown in Fig.~\ref{f69}(a).}

\textcolor{black}{ Fig.~\ref{f69}(c),(d) presents the results of the
bifurcation analysis for other values of the cycle parameters:
$T=65.45$, $ \lambda=4.73$ and
$0.1505<a_3<0.612$. As pointed out earlier, the stability of the
fixed point $\mathcal{O}$ (1-cycle) is determined by  $F'(z_0)$ and $\Phi'(z_0)$.}

\textcolor{black}{ 
Introduce  $\tau$  as
\begin{gather*}
 \tau=1/|\Lambda|,\quad \Lambda=\ln\:r_0,\\ 
r_0=\max_{1\leqslant i\leqslant 3}|\rho_i|.
\end{gather*}
The value of $\tau$ characterizes the convergence
time of the trajectory initiated a point in the basin  of
attraction of the stable fixed point $\mathcal{O}$ to the corresponding orbit. }

\textcolor{black}{  Fig.~\ref{f75}(a),(b) show variation of the
$\tau$  and $\rho_2$ in the intervals $-0.6<F'<0.0$ and  $\Phi'=-\dfrac{k_2}{k_4}F'$
for $a_3 = 0.3005$ and $a_3=0.2505$ ($T=66.75$,
$\lambda=4.66$), respectively.  
}

\begin{figure}[h]
\centering \centering
\begin{subfigure}[t]{0.85\columnwidth}  
            \centering 
            \includegraphics[width=0.85\columnwidth]{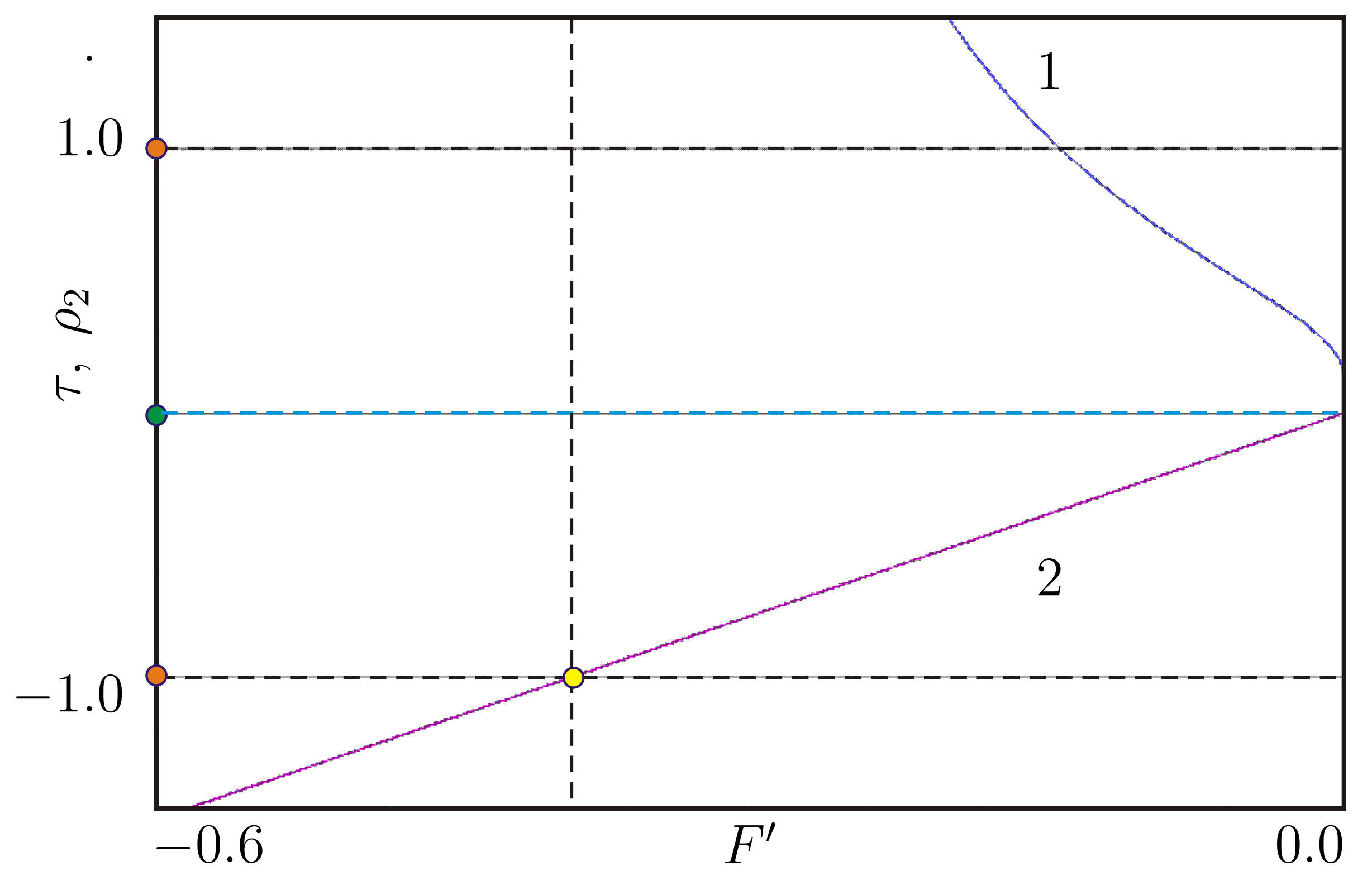}
        \end{subfigure}
        \hfill
        \begin{subfigure}[t]{0.85\columnwidth}  
            \centering 
            \includegraphics[width=0.85\columnwidth]{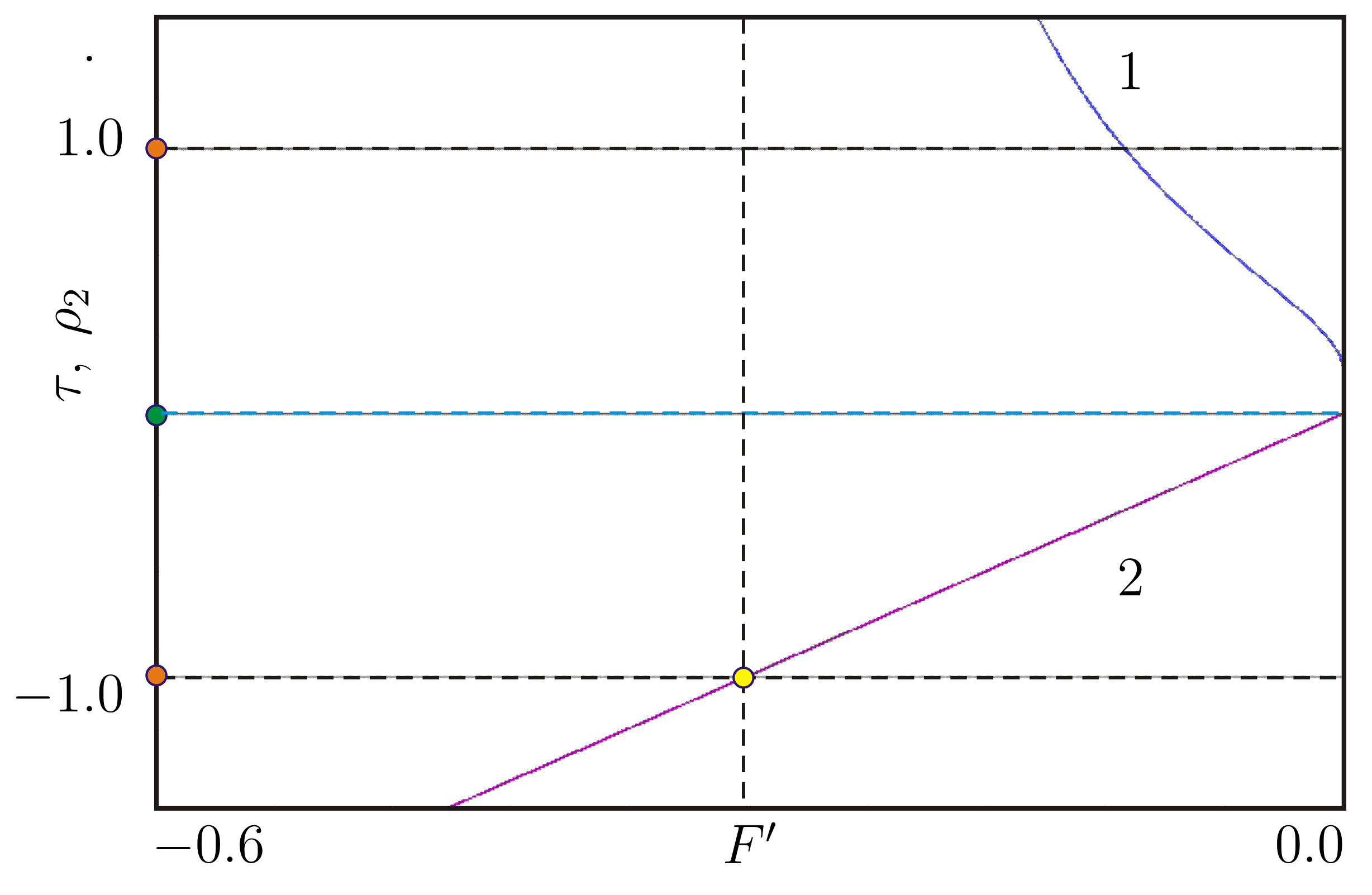}
        \end{subfigure}
\caption{\label{f75}\small (a) Variation  of  $\tau$ and $\rho_2$ on $-0.6<F'<0.0$ and $\Phi'=-\dfrac{k_2}{k_4}F'$ for $a_3=0.3005$0. (b) Variation of  $\tau$ and $\rho_2$ on $-0.6<F'<0.0$,  $\Phi'=-\dfrac{k_2}{k_4}F'$ for $a_3=0.2505$0. Here 1 denotes the  convergence time $\tau$ and 2 marks $\rho_2$. $T=66.7502$, $\lambda=4.6625$.} 
\end{figure}


\section{CONCLUSIONS}\label{sec.concl}
A novel problem of designing the IGO to admit a pre-defined periodic solution is introduced. It is exemplified by a case of stable 1-cycle with pre-defined solution parameters. It is demonstrated that the 1-cycle specifications are translated into a unique positive fixed point of the impulse-to-impulse discrete map. This fixed point can be rendered stable by selecting the modulation functions of the IGO. Further analysis is needed to control the type (monotonous, non-monotonous) and the speed of convergence  of the IGO solutions to the orbit corresponding to the obtained fixed point.
    
\vskip0.1cm

\bibliography{refs,observer}
\bibliographystyle{IEEEtran}
\end{document}